\documentclass[journal,10pt]{IEEEtran}
\usepackage{amsfonts}
\usepackage{amssymb}
\usepackage{amsmath}
\usepackage{mathrsfs}
\usepackage{verbatim}
\usepackage{subfigure}
\usepackage{balance}
\usepackage{booktabs}
\usepackage{fancybox}
\usepackage{bm}
\usepackage{changepage}
\usepackage{extarrows}
\usepackage{algorithm}
\usepackage{algorithmic}
\usepackage{diagbox}
\usepackage{multirow}
\usepackage{array}
\usepackage{graphicx}
\usepackage{epstopdf}
\usepackage{caption}
\newtheorem{Th}{Theorem}

\begin{document}

\title{\LARGE Opportunistic Jamming for Enhancing
Security: Stochastic Geometry
Modeling and Analysis}

\author{\IEEEauthorblockN{Chao Wang and Hui-Ming Wang,~\IEEEmembership{Senior Member, IEEE}}
\thanks{``Copyright (c) 2015 IEEE. Personal use of this material is permitted. However, permission to use this material for any other purposes must be obtained from the IEEE by sending a request to pubs-permissions@ieee.org.'' The authors are with the School of Electronic and Information Engineering,
Xi'an Jiaotong University, Xi'an, 710049, Shaanxi, China. Email:
{\tt wangchaoxuzhou@stu.xjtu.edu.cn} and {\tt xjbswhm@gmail.com}.
The contact author is Hui-Ming Wang. The work was supported by the Foundation for the Author of National Excellent Doctoral Dissertation of China under Grant 201340, the National High-Tech Research and Development Program of China under Grant No.
2015AA011306, the New Century Excellent Talents Support Fund of China under Grant NCET-13-0458, the Fok Ying Tong Education Foundation under Grant 141063, and the Fundamental Research Funds for the Central University under Grant No. 2013jdgz11.
}
}

\markboth{IEEE Transactions on Vehicular Technology,~Vol.~XX, No.~XX, XXX~2016}
{}

\maketitle

\begin{abstract}
This correspondence studies the secrecy communication of the single-input single-output  multi-eavesdropper (SISOME) channel with multiple single-antenna jammers, where the jammers and  eavesdroppers are distributed according
to the independent two-dimensional homogeneous Poisson
point process (PPP). For enhancing the physical layer security, we propose an opportunistic multiple jammer selection scheme, where the jammers whose channel gains to the legitimate receiver less than a threshold, are selected to transmit independent and identically distributed (\emph{i.i.d.}) Gaussian jamming signals to confound the eavesdroppers.
We characterize the secrecy throughput achieved by our proposed jammer selection scheme, and  show that the secrecy throughput is a quasi-concave  function of the selection threshold.
\end{abstract}


\IEEEpeerreviewmaketitle

\section{Introduction}
Artificial jamming scheme introduced in \cite{Goel:TWC08} has been recognized as an active approach for improving the physical layer security.
For improving the security of rely networks, cooperative jamming has been introduced in \cite{Dong,ZhengGan}, where multiple single-antenna jammers transmit  jamming signals collaboratively to confound the eavesdropper.
In \cite{Profwangadd4}, a hybrid opportunistic relaying and jamming scheme has been proposed. In \cite{DoctorWang1}, exploring the heterogeneous large-scale fading, a distributed jamming scheme is proposed for securing the single-input multi-output  transmission.
In \cite{CjZhou}, \cite{DoctorWangTWC2}, cooperative jamming has been applied in random networks under the framwork of stochastic geometry.
A survey of the recent advances on cooperative jamming for enhancing security can be found in \cite{Profwangadd5}.



Although cooperative jamming can secure the single antenna communication efficiently, for its implementation, the beamformer weights multiplied at each jammer should be designed coordinately, which will result in a high overhead, especially when these jammers are spatially separated in a random network.  Besides, a common jamming signal vector should be shared among multiple jammers in the cooperative jamming, and the shared jamming
signals should be secured against  eavesdroppers, which would increase the design complexity further.
To reduce the system complexity, a fully distributed jamming scheme without centralized design is more preferable in a practical system.
In \cite{DoctorWangTWC2}, we  proposed a distributed opportunistic jammer selection scheme to secure a single-input multiple-output (SIMO) transmission, where each cooperation node takes its channel direction information (CDI) of the legitimate channel as the jammer selection criteria, and transmits independent jamming signals.

In this correspondence, we propose an opportunistic jammer selection scheme for securing the single-input single-output  multi-eavesdropper (SISOME) wiretap channel in a random network, where positions of multiple jammers and eavesdroppers as two independent Poisson point processes (PPPs).
In particular, multiple single-antenna jammers whose channel gains to the legitimate receiver are less than a fixed selection threshold, are selected to transmit \emph{i.i.d.} jamming signals. We characterize its achievable secrecy throughput and obtain the global optimal selection threshold.
Different from cooperative jamming schemes proposed in \cite{Dong, Profwangadd4} which require multiple jammers to transmit jamming signals collaboratively and share jamming signals, our proposed opportunistic jammer selection scheme makes it possible that multiple jammers can send jamming signals independently and do not need to share jamming signals.
The proposed scheme is also different from the uncoordinated one in \cite{DoctorWangTWC2} since we
take the channel gain information (CGI) of jammers into consideration rather than the CDI, so that it
can be applied for securing the SISOME wiretap channel.
%

 Specially, our contributions can be summarized as follows:

1) Under a stochastic geometry framework, an efficient multiple jammer selection scheme is proposed for securing the SISOME wiretap channel, where each jammer transmits jamming signals in an uncoordinated way so that a very low system overhead is consumed.

2) Compact analysis results of the achievable connection outage and the secrecy outage are derived, which facilitates the numerical evaluation of the  secrecy throughput achieved by our proposed scheme.

3)  We prove that the secrecy throughput is a quasi-concave function of the selection threshold, which can be calculated numerically with a low complexity.

\emph{Notation:} $\mathbf{x} \sim \mathcal{CN}\left(\mathbf{\Lambda}, \mathbf{\Delta}\right)$ denotes the circular symmetric complex Gaussian vector with mean vector $\mathbf{\Lambda}$ and covariance matrix $\mathbf{\Delta}$, $y \sim \textrm{exp}(1)$ denotes the exponential random variable with the parameter 1, $\textrm{ln}(\cdot)$ denotes the base-e logarithm function,
$\Gamma(x)$ is the gamma function \cite[eq. (8.310)]{Table}, $\gamma(x,y)$ is the lower incomplete gamma function \cite[8.350.1]{Table}, $\textrm{Ei}(x)$ is the exponential integral function \cite[8.211.1]{Table}, $\mathbb{E}(\cdot)$ denotes the statistical expectation.

\section{System Model and Assumptions}
\subsection{System model}
We consider a wiretap channel consisting of a legitimate transmitter (Alice), a legitimate receiver (Bob), multiple jammers, and  multiple eavesdroppers (Eves).
All the nodes in the considered system are equipped
with a single antenna\footnote{Due to size, cost, or hardware limitations, in some wireless systems, e.g., ad hoc network and D2D communication scenarios, the nodes can not support multiple antennas. Therefore, our proposed opportunistic jammer selection is suitable for enhancing the physical layer security of these systems.}.
For guaranteeing security, we proposed an opportunistic multiple jammer selection scheme, where multiple jammers are selected in a distributed fashion to cover the secrecy transmission by sending independent artificial jamming signals.

We consider both large-scale and small-scale fading for wireless channels. For the large-scale fading, we adopt the standard path loss model $l(r)=r^{-\alpha}$, where $r$ denotes the distance and $\alpha>2$ is the fading exponent \cite{RandomGraphs}. For the small-scale fading, just as \cite{CjZhou},  we assume independent quasi-static Rayleigh fading, and the channel gains  follow the exponential distribution with the parameter 1.
Since Eves are passive, their instantaneous CSIs and locations are unavailable. We assume that the locations of jammers and Eves are modeled by two independent homogeneous PPPs on $\mathbb{R}^2$ with the densities $\lambda_J$ and $\lambda_E$, respectively. Such random PPP model is well motivated
by the random and unpredictable locations of eavesdroppers.
We denote the location set of all jammers, selected jammers and eavesdroppers as $\Phi_J$, $\Phi_J^s$  and $\Phi_E$, respectively, and the distance between Alice and Bob as $d$.

When multiple selected jammers are activated to transmit jamming signals independently, the received confidential signals at Bob would be disturbed by jamming signals, and the received signal to interference plus noise ratio (SINR) at Bob can be calculated as
$
\textrm{SINR}_B=\frac{P_Sh_Bd^{-\alpha}}{P_J\sum_{i\in\Phi_J^s}g_{iB}D_{iB}^{-\alpha}+N_0},
$
where $P_S$ is the transmit power of Alice, $P_J$ is the transmit power of each jammer, $h_B\sim\textrm{exp}(1)$ is the channel gains between Alice and Bob, $g_{iB}\sim\textrm{exp}(1)$ and $D_{iB}$ are the channel gain and the distance between the $i$th jammer and Bob, $N_0$ is the noise power received at Bob.

Since the noise power at Eve is typically unknown to Alice, we adopt a conservative approach, as done in \cite{Profwangadd4,DoctorWangTWC2,CjZhou}, to design the secure transmission scheme by assuming that the noise power at Eve is zero.
In such case, the received signal to interference ratio (SIR) at the $j$th eavesdropper can be calculated as
$
\textrm{SIR}_{E_j}=\frac{P_Sh_{E_j}d_{E_j}^{-\alpha}}{P_J\sum_{i\in\Phi_J^s}g_{iE_j}D_{iE_j}^{-\alpha}},
$
where $h_{E_j}\sim\textrm{exp}(1)$ and $d_{E_j}$ are the channel gains and the distance between Alice and the $j$th Eve, $g_{iE_j}\sim\textrm{exp}(1)$ and $D_{iE_j}$ are the channel gain and the distance between the $i$th jammer and the $j$th Eve.

\subsection{Opportunistic Jammer Selection}
For securing the legitimate transmission, an opportunistic jammer selection is performed to maximize the receiving performance difference between the legitimate receiver and eavesdroppers. Since the locations of eavesdroppers are unavailable,
it is difficult to select the jammers according to their locations.
In this correspondence, we perform the jammer selection only according to the channel gains $g_{iB}$, without considering the large-scale fading effects.
In particular,  multiple jammers
whose channel gains to Bob less than a fixed threshold, are selected to transmit \emph{i.i.d.} artificial jamming signals. With such jammer selection scheme, the transmitted jamming signals would confound multiple Eves while keeping Bob as non-intrusive as possible.
Then, the channel gains of the selected jammers to Bob should be in the following set
\begin{align}
\mathcal{R}_J=\left\{g_{iB}|\quad g_{iB}\leq \delta\textrm{ and  } i\in {\Phi}_J\right\},\label{RJ}
\end{align}
where $\delta$ is the selection threshold.
Although a smaller $\delta$ would result in less jamming signals received at Bob, the number of competent jammers also decreases. Therefore, there is a tradeoff between decreasing the harmful interference received at Bob and increasing the power of the jamming signals received at Eves. Therefore, $\delta$ should be optimized to maximize the achievable secrecy performance.

With such jammer selection scheme, the channel gain ${{g}}_{iB}$ should lie in the set $\mathcal{R}_J$ and  we define the probability of ${{g}}_{iB}\in\mathcal{R}_J$ as $\textrm{Prob}_J$, which can be calculated as
$
\textrm{Prob}_J = 1-\textrm{exp}(-\delta).
$

Therefore, the random variable  ${{g}}_{iB}$ given ${{g}}_{iB}\in \mathcal{R}_J$ has the conditional density function
\begin{align}
f_{ {{g}}_{iB}}\left(x| {{g}}_{iB}\in\mathcal{R}_J\right) =\frac{\textrm{exp}(-x)1_{[0,\delta]}(x)}{1-\textrm{exp}(-\delta)} ,
\end{align}
where $1_{[0,\delta]}(x)$ is the indicator function having 1 if  $x\in [0,\delta]$ and 0 otherwise.


According to \cite{RandomGraphs}, we know that the resulting selected jammer set is a thinning of the homogeneous PPP of the intensity $\lambda_J$ with the retention probability $\textrm{Prob}_J$. Then the resulting  selected jammer set is a homogeneous PPP $\Phi^s_J$ of the intensity $\lambda_J^s$ which is given by
$
\lambda_J^s=\left(1-\textrm{exp}\left(-\delta\right)\right)\lambda_J
$.
From $\lambda_J^s$, we can find that $\lambda_J^s$ becomes smaller as $\delta$ decreases.

\subsection{Secrecy Throughput}
In this correspondence, we study the achievable secrecy performance of our proposed opportunistic jammer selection scheme by considering  the outage based secrecy metrics.

In the following, we denote the confidential message rate as $R_s$ and the rate of the transmitted codeword as $R_t$. When the capacity of the channel from Alice to Bob is below the transmission rate $R_t$, Bob can not decode the received message correctly. The probability of this event is defined as \emph{connection outage probability}. When the maximal capacity of the channels from Alice to multiple Eves is above the rate $R_e\triangleq R_t-R_s$, the confidential information can not be perfectly secured against eavesdropping.
The probability of such event is defined as\emph{ secrecy outage probability}.
Under a given connection outage probability $\sigma$ and secrecy outage probability $\epsilon$,  the secrecy throughput $\mu$ is defined as
\begin{align}
\mu \triangleq (1-\sigma)R_s,\label{secrecyThroughputDefinition}
\end{align}
which is suitable for evaluating the secrecy performance of systems with stringent delay
constraints.

\section{Secrecy Throughput Analysis and Optimization}
In this section, we firstly study the achievable secrecy throughput of the proposed jammer selection scheme. Then, we prove that the achievable secrecy throughput is a quasi-concave function of $\delta$, and the optimal $\delta$ can be located efficiently by many numerical methods.

\subsection{Secrecy Throughput Analysis}
Defining the SINR threshold  for the connection outage as $\beta_B$, the corresponding rate threshold is $R_t=\textrm{log}_2\left(1+\beta_B\right)$. The connection outage, $p_{co}$ is defined as $p_{co}\triangleq\textrm{Prob}\left(\textrm{SINR}_B\leq \beta_B\right)$, whose analysis result is given by the following theorem.
 \begin{Th}
 Setting $\rho=\frac{2}{\alpha}$, $p_{co}$ can be calculated as
\begin{align}
&p_{co}=1-\textrm{exp}\left(-\frac{N_0d^{\alpha}\beta_B}{P_s}\right.
\nonumber\\
&\qquad\quad\left.-\lambda_J\pi\gamma\left(\rho+1,\delta\right)\Gamma\left(1-\rho\right)\left(\frac{d^{\alpha}P_J}{P_s}\right)^{\rho}\beta_B^{\rho}\right).\label{pco}
 \end{align}
 \end{Th}
\begin{proof}
\begin{align}
p_{co}=&\textrm{Prob}\left(h_B\leq \frac{\beta_B\left(N_0+P_J\sum_{i\in\Phi_J^s}D_{iB}^{-\alpha}g_{iB}\right)d^{\alpha}}{P_s}\right)
\nonumber\\
&=1-\textrm{exp}\left(-\frac{\beta_BN_0d^{\alpha}}{P_S}\right)
\nonumber\\
&\qquad\quad\mathbb{E}\left(\textrm{exp}\left(-\frac{d^{\alpha}P_J\beta_B}{P_s}\sum_{i\in{\Phi}_J^s}D_{jB}^{\alpha}g_{iB}\right)\right).
\end{align}
Then, employing \cite[eq. (8)]{RandomGraphs}, the expectation $\mathbb{E}\left(\textrm{exp}\left(-\frac{d^{\alpha}P_J\beta_B}{P_s}\sum_{j\in{\Phi}_J^s}D_{jB}^{\alpha}g_{jB}\right)\right)$ can be calculated as
\begin{align}
&\mathbb{E}\left(\textrm{exp}\left(-\frac{d^{\alpha}P_J\beta_B}{P_s}\sum_{j\in{\Phi}_J^s}D_{jB}^{\alpha}g_{jB}\right)\right)
=\nonumber\\
&\textrm{exp}\left(-\lambda_J^s\pi\mathbb{E}\left(g_{jB}^{\rho}\right)\Gamma\left(1-\delta\right)\left(\frac{d^{\alpha}P_J\beta_B}{P_s}\right)^{\rho}\right),
\label{expectation2}
\end{align}
and
\begin{align}
\mathbb{E}\left(g_{iB}^{\rho}\right)=\int^{\delta}_0\frac{x^{\rho}\textrm{exp}(-x)}{1-\textrm{exp}\left(-\delta\right)}dx=\frac{\gamma\left(\rho+1,\delta\right)}{1-\textrm{exp}(-\delta)}.
\label{expectation}
\end{align}
Then substituting (\ref{expectation}) into (\ref{expectation2}), (\ref{pco}) can be derived.
\end{proof}

Assuming that the signal-to-interference ratio (SIR) threshold for the secrecy outage as $\beta_E$, the corresponding rate threshold is $R_e=\textrm{log}_2\left(1+\beta_E\right)$. The secrecy outage is
\begin{align}
p_{so}=\textrm{Prob}\left(\max_{j\in\Phi_E}\textrm{SIR}_{E_j}\geq \beta_E\right),
\end{align}
whose analysis result is given by the following theorem.
\begin{Th}
$p_{so}$ can be calculated as
\begin{align}
&p_{so}=
1-
\nonumber\\
&\textrm{exp}\left(-\frac{\lambda_E}{\left(\frac{P_J\beta_E}{P_S}\right)^{\rho}\lambda_J\left(1-\textrm{exp}(-\delta)\right)\Gamma\left(1+\rho\right)\Gamma\left(1-\rho\right)}\right).\label{secrecyoutageNonoise}
\end{align}
\end{Th}
\begin{proof}
Following the probability generating functional (PGFL) \cite{RandomGraphs}, we have
\begin{align}
&p_{so}=1-\mathbb{E}\left(\prod_{j\in\Phi_E}\left(\textrm{Prob}\left(\textrm{SIR}_{E_j}\leq \beta_E\right)\right)\right)
\nonumber\\
&=1-\textrm{exp}\left(-\lambda_E\int_{\mathbb{R}^2}\left(1-\textrm{Prob}\left(\textrm{SINR}_{E_j}\leq \beta_E\right)d\mathbf{x}_{E_j}\right)\right),\label{PsoDerivation}
\end{align}
where $\mathbf{x}_{E_j}$ denotes the location of the $j$th eavesdropper, and $\textrm{Prob}\left(\textrm{SIR}_{E_j}\leq \beta_E\right)$ denotes the probability that the received SIR at the $j$th eavesdropper is less than $\beta_E$. Assuming that the distance between the $j$th eavesdropper and Alice is $d_{E_j}$,
\begin{align}
&\textrm{Prob}\left(\textrm{SIR}_{E_j}\leq \beta_E\right)
=\nonumber\\
&1-\mathbb{E}\left(\textrm{exp}\left(-\frac{P_Jd^{\alpha}_{E_j}\sum_{i\in\Phi_J^s}g_{iE_j}D_{iE_j}^{-\alpha}}{P_S}\right)\right)
\overset{(m)}{=}1-
\nonumber\\
&\textrm{exp}\left(-\lambda_J\left(1-\textrm{exp}(-\delta)\right)\Gamma(1+\delta)\Gamma(1-\delta)d_{E_j}^{2}\left(\frac{\beta_EP_J}{P_S}\right)^{\rho}\right)
\label{ProbabilityJthEavesdropper}
\end{align}
Step $(m)$ is due to \cite[eq. (8)]{RandomGraphs}.
Then, substituting (\ref{ProbabilityJthEavesdropper}) into (\ref{PsoDerivation}) and changing to a polar coordinate system, we have
\begin{align}
&p_{so}=1-\textrm{exp}\left(-\lambda_E\pi\int^{+\infty}_0\textrm{exp}\left(-\Psi y\right)dy\right),\label{additionalEquation}
\end{align}
where $\Psi\triangleq\left(\frac{P_J\beta_E}{P_S}\right)^{\rho}\lambda_J\left(1-\textrm{exp}(-\delta)\right)\pi\Gamma(1+\rho)\Gamma(1-\rho)$.
After completing the integral, the proof can be completed.

\end{proof}

Then with Theorem 2, the required rate redundancy $R_e$ for maintaining the secrecy outage constraint $p_{so}\leq \epsilon$,
can be calculated by setting $p_{so}=\epsilon$, which is given by
\begin{align}
&R_e =
\textrm{log}_2\Bigg(1+
\nonumber\\
&\left.\left(\frac{\lambda_E}{\left(\frac{P_J}{P_S}\right)^{\rho}\lambda_J\left(1-\textrm{exp}(-\delta)\right)\Gamma(1+\rho)\Gamma(1-\rho)\textrm{ln}\frac{1}{1-\epsilon}}\right)^{\frac{\alpha}{2}}\right).\label{RERE}
\end{align}

Accordingly, the maximal $R_t$ can be calculated from the connection outage constraint $p_{co}\leq \sigma$ in (\ref{pco}). Unfortunately, the closed-form analysis result of the maximal $R_t$ can not be obtained from (\ref{pco}). But since $p_{co}$ is a monotonic increasing function of $R_t$, employing the bisection search, the numeral result of the maximal $R_t$ can be obtained.
Then, according to (\ref{secrecyThroughputDefinition}), the secrecy throughput $\mu$ can be calculated by
\begin{align}
\mu = (R_t-R_e)(1-\sigma)\label{secrecythroughputformula}.
\end{align}

In the following subsection, we would optimize the selection threshold to maximize the achievable secrecy throughput $\mu$.

\subsection{Optimizing $\delta$ for Secrecy Throughput Maximization}
From the discussions above, we can find that there is an optimal tradeoff between protecting Bob from the harmful interference and increasing the jamming power received at each Eve, which is determined by $\delta$. The following theorem shows that the achievable secrecy throughput is a quasi-concave function of $\delta$.
\begin{Th}
$\mu$ is a quasi-concave function $\delta$.
\end{Th}
\begin{proof}
For maximizing $\mu$, the maximal $\beta_B$ and minimal $\beta_E$ should make the constraints above active, which leads to $p_{co}=\sigma$ and $p_{so}=\epsilon$.

For notational conciseness, we denote $b\triangleq\frac{N_0d^{\alpha}}{P_S}$, $a =\lambda_J\pi\Gamma\left(1-\rho\right)\left(\frac{d^{\alpha}P_J}{P_S}\right)^{\rho}$, and
$c\triangleq \frac{\lambda_EP_S^{\rho}}{P_J^{\rho}\lambda_J\Gamma\left(1+\rho\right)\Gamma\left(1-\rho\right)}$.
Then, from $p_{co}=\sigma$ and $p_{so}=\epsilon$, according to implicit function theorem, we have
\begin{align}
&\frac{d\beta_B}{d\delta}=-\frac{a\delta^{\rho}e^{-\delta}\beta_B^{\rho}}{b+a\rho\gamma\left(\rho+1,\delta\right)\beta_B^{\rho-1}}
,\label{Revisedfirst1}\\
&\frac{d\beta_E}{d\delta}=-\frac{\beta_E\textrm{exp}(-\delta)}{\rho\left(1-\textrm{exp}(-\delta)\right)}.\label{first2}
\end{align}
Since $\mu=\left(1-\sigma\right)\left(R_t-R_e\right)$, where only $R_t-R_e$ is determined by $\delta$, $\mu$ is a quasi-concave function of $\delta$, if and only if $R_t-R_e$ is a quasi-concave function of $\delta$.

In the following, we show that $R_t-R_e$ satisfies the second-order conditions of the quasi-concave function \cite[Section 3.4.3]{ConvexOptimization}, which is given as follows. $f(x)$ is a quasi-concave function on $\mathbb{R}$, if and only if
$
\frac{df(x)}{dx}=0\Rightarrow\frac{d^2f(x)}{d^2x}\leq 0.
$

When $\frac{d\left(R_t-R_e\right)}{d\delta}=0$, we have
\begin{align}
&\frac{a\delta^{\rho}e^{-\delta}\beta_B^{\rho}}{\left(1+\beta_B\right)\left(b+a\rho\gamma\left(\rho+1,\delta\right)\beta_B^{\rho-1}\right)}
\nonumber\\
&=\frac{\beta_E\textrm{exp}(-\delta)}{\left(1+\beta_E\right)\rho\left(1-\textrm{exp}(-\delta)\right)}.\label{firstorder}
\end{align}
\begin{figure*}[!t]
\begin{align}
&\frac{d^2\left(R_t-R_e\right)}{d^2\delta}=\underset{T_1}{\underbrace{e^{-\delta}\left(\frac{a\delta^{\rho}e^{-\delta}\beta_B^{\rho}}{b+a\rho\gamma\left(\rho+1,\delta\right)\beta_B^{\rho-1}}
-\frac{\beta_E\textrm{exp}(-\delta)}{\rho\left(1-\textrm{exp}(-\delta)\right)}\right)}}
-e^{-\delta}\left(\underset{T_2}{\underbrace{\frac{\beta_Ee^{-\delta}+\beta_E\rho e^{-\delta}+\beta^2_E\rho e^{-\delta}}{\left((1+\beta_E)\rho\left(1-e^{-\delta}\right)\right)^2}}}\right.
\nonumber\\
&\ \ \ \left.\underset{T_3}{\underbrace{+\frac{\Theta+a^2\delta^{2\rho}e^{-\delta}\beta_B^{2\rho}-a^2\rho\delta^{2\rho}e^{-\delta}\beta_B^{2\rho-1}-a^2\rho\delta^{2\rho}e^{-\delta}\beta_B^{2\rho}-a^2\rho e^{-\delta}\delta^{2\rho}\beta_B^{2\rho-1}-a^2\rho e^{-\delta}\delta^{2\rho}\beta_B^{2\rho}}{\left((1+\beta_B)\left(b+a\rho\gamma\left(\rho+1,\delta\right)\beta_B^{\rho-1}\right)\right)^2}
}}\right),\label{secondorderExprression}
\end{align}
\hrulefill
\end{figure*}

Then, in the following, we prove that when $\delta$ satisfies (\ref{firstorder}), the second-order derivative
$
\frac{d^2\left(R_t-R_e\right)}{d^2\delta}\leq 0.
$

From (\ref{Revisedfirst1}) and (\ref{first2}), we first derive the second-order derivative of $R_t-R_e$ with respect to $\delta$ after tedious manipulations, given by (\ref{secondorderExprression}) at the top of this page,
where
\begin{align}
\Theta&=a\rho\delta^{\rho-1}\beta_B^{\rho-1}\left(1+\beta_B\right)\left(\beta_B b+a\rho\gamma(\rho+1,\delta)\beta_B^{\rho}\right)
\nonumber\\
&+\left(1+\beta_B\right)\frac{a^3\rho\gamma(\rho+1,\delta)}{b+a\rho\gamma\left(\rho+1,\delta\right)\beta_B^{\rho-1}}>0.
\nonumber
\end{align}

When $\delta$ satisfies (\ref{firstorder}), $T_1$ in (\ref{secondorderExprression}) is zero. In the following, we concentrate on proving that $T_2+T_3$ in (\ref{secondorderExprression}) is positive.
Before proceeding, we bound $T_2$ with the following procedures.

From (\ref{firstorder}), we have
\begin{align}
&\frac{1}{\left(\left(1+\beta_E\right)\rho\left(1-\textrm{exp}(-\delta)\right)\right)^2}
\nonumber\\
&=\frac{a^2\delta^{2\rho}\beta_B^{2\rho}}{\beta_E^2\left(1+\beta_B\right)^2\left(b+a\rho\gamma\left(\rho+1,\delta\right)\beta_B^{\rho-1}\right)^2}.\nonumber
\end{align}
Then with the equation above, $T_2$ can be bounded as (\ref{addLowerBoundT2}) at the top of the next page,
\begin{figure*}
\begin{align}
T_2&=\frac{a^2\delta^{2\rho}e^{-\delta}\beta_B^{2\rho}+\rho a^2\delta^{2\rho}e^{-\delta}\beta_B^{2\rho}}{\beta_E\left(1+\beta_B\right)^2\left(b+a\rho\gamma\left(\rho+1,\delta\right)\beta_B^{\rho-1}\right)^2}
+
\frac{a^2\rho\delta^{2\rho}e^{-\delta}\beta_B^{2\rho}}{\left(1+\beta_B\right)^2\left(b+a\rho\gamma\left(\rho+1,\delta\right)\beta_B^{\rho-1}\right)^2}
\nonumber\\
&\qquad \overset{(e)}{\geq}\frac{a^2\delta^{2\rho}e^{-\delta}\beta_B^{2\rho-1}+\rho a^2\delta^{2\rho}e^{-\delta}\beta_B^{2\rho-1}}{\left(1+\beta_B\right)^2\left(b+a\rho\gamma\left(\rho+1,\delta\right)\beta_B^{\rho-1}\right)^2}
+
\frac{a^2\rho\delta^{2\rho}e^{-\delta}\beta_B^{2\rho}}{\left(1+\beta_B\right)^2\left(b+a\rho\gamma\left(\rho+1,\delta\right)\beta_B^{\rho-1}\right)^2},\label{addLowerBoundT2}
\end{align}
\hrulefill
\end{figure*}
where step $(e)$ holds since the achievable secrecy throughput is nonnegative.

Then $T_2+T_3$ can be bounded as
\begin{align}
T_2+T_3\geq\frac{\Theta+\Delta_1+\Delta_2
}{\left(1+\beta_B\right)^2\left(b+a\rho\gamma\left(\rho+1,\delta\right)\beta_B^{\rho-1}\right)^2},
\end{align}
where
$\Delta_1\triangleq a^2\delta^{2\rho}e^{-\delta}\beta_B^{2\rho}-a^2\rho\delta^{2\rho}e^{-\delta}\beta_B^{2\rho}$,
$\Delta_2\triangleq
a^2\delta^{2\rho}e^{-\delta}\beta_B^{2\rho-1}-a^2\rho e^{-\delta}\delta^{2\rho}\beta_B^{2\rho-1}$.

Since $\rho<1$, we have $\Delta_1,\Delta_2>0$ and $T_2+T_3>0$. Then we can conclude that when $\delta$ satisfies (\ref{firstorder}), $\frac{d^2\left(R_t-R_e\right)}{d^2\delta}<0$. Therefore  $R_t-R_e$ is a quasi-concave function of $\delta$, and $\mu$ is a quasi-concave function of $\delta$.
\end{proof}


Since $\mu$ is a quasi-concave function of $\delta$, the optimal $\delta$ for maximizing $\mu$ can be located by the bisection search algorithm in \cite{ConvexOptimization}. For implementing the bisection search algorithm, we need
the first-order derivative of $R_t-R_e$ with respect to $\delta $, which can be derived  from (\ref{Revisedfirst1}) and (\ref{first2}). In particular,
\begin{align}
&\frac{d\left(R_t-R_e\right)}{d\delta}=-\frac{a\delta^{\rho}e^{-\delta}(\beta^o_B)^{\rho}}{\left(1+(\beta^o_B)\right)\left(b+a\rho\gamma\left(\rho+1,\delta\right)(\beta^o_B)^{\rho-1}\right)}
\nonumber\\
&+\frac{\beta^o_E\textrm{exp}(-\delta)}{\left(1+\beta^o_E\right)\rho\left(1-\textrm{exp}(-\delta)\right)},\label{FirstDerivateEvaluation}
\end{align}
where $\beta_B^o$ is the maximal $\beta_B$ satisfying $p_{co}=\sigma$, and  $\beta_E^o$ is the minimal $\beta_E$ satisfying $p_{so}=\epsilon$. The closed-form analysis result of $\beta_E^o$ can be derived from (\ref{secrecyoutageNonoise}), which is given by $\beta_E^o=\left(\frac{\lambda_E}{\left(\frac{P_J}{P_S}\right)^{\rho}\lambda_J\left(1-\textrm{exp}(-\delta)\right)\Gamma(1+\rho)\Gamma(1-\rho)\textrm{ln}\frac{1}{1-\epsilon}}\right)^{\frac{\alpha}{2}}$. The numerial result of $\beta_B^o$ can be derived from (\ref{pco}) by numerical methods.
\section{Simulation Results and Discussions}
Setting the distance between Alice and Bob $d=1$ m, the path loss exponent $\alpha=3$, considering different $\lambda_E$ and $\lambda_J$, some representative simulation results are provided to gain more insights into the proposed
jammer selection scheme. In Fig. \ref{secrecythroughputwithlambdaE}, we show the secrecy performance comparison between our proposed opportunistic jammer selection scheme and the random jammer selection scheme in \cite{CjZhou}, where multiple randomly selected jammers transmit jamming signals independently to interfere with eavesdroppers. Without jammer selection, the jamming signals transmitted from multiple randmly selected jammers would deterioriate the receive performance of Bob and Eves simultaneously. Therefore, compared with our proposed jammer selection scheme, the secrecy performance deterioration of the random jammer selection scheme can be anticipated. From Fig. \ref{secrecythroughputwithlambdaE}, we can find that compared with the ramdom jammer selection scheme, the secrecy performance improvement acheived by our proposed opportunistic jammer selection scheme is substantial, which validates the efficiency of our proposed secure transmission scheme.
\begin{figure}[!t]
\centering
\includegraphics[width=2.5in]{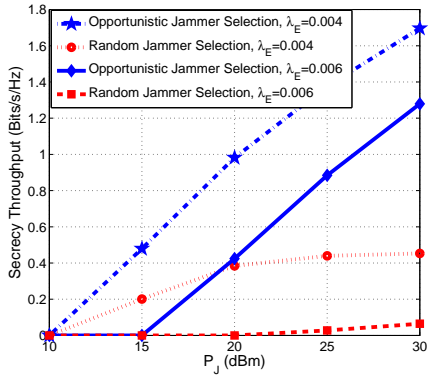}
\caption{Performance comparison between our scheme and random jammer selection proposed in \cite{CjZhou} for $\sigma=0.1,\epsilon=0.01,P_S=20$ dBm.}
\label{secrecythroughputwithlambdaE}
\end{figure}


\begin{figure}[!t]
\centering
\includegraphics[width=2.5in]{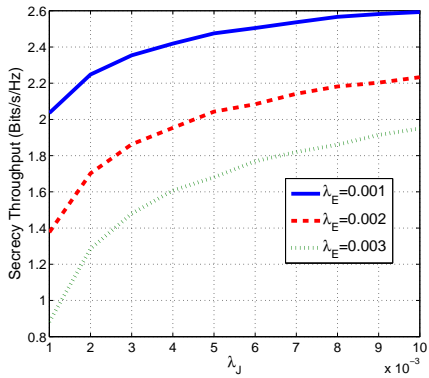}
\caption{Average secrecy throughput achieved by our  scheme vs $\lambda_J$ for $\sigma=0.1,\epsilon=0.01,P_S=20$ dBm, and $P_J=30$ dBm.}
\label{secrecythroughputwithLambdaJ}
\end{figure}

In Fig. \ref{secrecythroughputwithLambdaJ}, we plot the   secrecy throughput versus $\lambda_J$, which shows that the achievable  secrecy throughput increases with the increasing $\lambda_J$, and there is a
diminishing return in the achievable secrecy throughput as $\lambda_J$ increases. This because that although the increasing alternative jammers would increase the interference power received at potential eavesdroppers, the interference received at Bob also increases. Furthermore, although the cost for the secrecy communication decreases with the increasing alternative jammers, the maximal $R_t$ is limited by the legitimate channel. Therefore, when the transmission power at Alice remains unchanged, the improvement of the secrecy throughput by increasing $\lambda_J$ is limited.

\section{Conclusion}
In this correspondence, we proposed an opportunistic jammer selection scheme for improving the security of the SISOME wiretap channel where jammers whose channel gains to Bob less than a threshold, are selected to transmit \emph{i.i.d.}  jamming signals to confound  eavesdroppers.
We characterize the secrecy throughput achieved by the scheme, and prove it is a quasi-concave function of the selection threshold.  Simulation results confirm the efficiency of our proposed secure transmission scheme.

%




\end{document}